%% file: main.tex
\documentclass[a4paper, 11pt]{article}

\usepackage{fullpage}
\usepackage[english]{babel}
\usepackage[latin1]{inputenc}
\usepackage{mathrsfs}
\usepackage{amsfonts}
\usepackage{amssymb}
\usepackage{amsmath}
\usepackage{amsthm}
\usepackage{graphicx}
\usepackage[usenames,dvipsnames]{color}
\usepackage[colorlinks=true,citecolor=blue,linkcolor=magenta]{hyperref}
\usepackage{microtype}
\usepackage{braket}
\usepackage{tabularx,colortbl}
\usepackage{pifont}
\usepackage{mathtools}
\usepackage{color}
\usepackage{dsfont}
\usepackage{algorithm}
\usepackage{chngcntr}
\usepackage[table, x11names]{xcolor}
\usepackage{epstopdf}
\usepackage{caption}
\usepackage{subfig}
\usepackage[titletoc,title]{appendix}
\usepackage{thmtools}
\usepackage{thm-restate}

\newcommand{\comment}[1]{}

\newcommand{\norm}[1]{\lVert#1\rVert}
\newcommand{\idty}[1]{\mathbb{1}}
\newcommand{\ovsqrt}[1]{\frac{1}{\sqrt{2}}}

\newcommand{\tr}[1]{\mathrm{Tr}}

\usepackage{algpseudocode}

\renewcommand{\tr}[1]{\mathrm{Tr}\left( #1 \right)}

\newtheorem{theorem}{Theorem}

\newtheorem{corollary}{Corollary}

\newtheorem{lemma}{Lemma}
\newtheorem{prop}[theorem]{Proposition}

\newtheorem{definition}{Definition}

\newcommand{\ketbra}[2]{\ket{#1}\!\bra{#2}}
\newcommand{\pnorm}[2]{\left\|#2\right\|_#1}
\newcommand{\maxnorm}[1]{\left\|#1\right\|_{\mathrm{max}}}

\newcommand{\polylog}{\mathrm{polylog}}
\newcommand{\poly}{\mathrm{poly}}

\title{A quantum algorithm for simulating non-sparse Hamiltonians}

\author{Chunhao Wang\thanks{Institute for Quantum Computing and School of Computer Science, University of Waterloo}
\thanks{Department of Computer Science, University of Texas at Austin} 
\and Leonard Wossnig\thanks{Department of Computer Science, University College London}}

\begin{document}

\date{\empty}
\maketitle
\begin{abstract}
  We present a quantum algorithm for simulating the dynamics of Hamiltonians that are not necessarily sparse. Our algorithm is based on the input model where the entries of the Hamiltonian are stored in a data structure in a quantum random access memory (qRAM) which allows for the efficient preparation of states that encode the rows of the Hamiltonian. We use a linear combination of quantum walks to achieve poly-logarithmic dependence on precision. The time complexity of our algorithm, measured in terms of the circuit depth, is $O(t\sqrt{N}\norm{H}\,\polylog(N, t\norm{H}, 1/\epsilon))$, where $t$ is the evolution time, $N$ is the dimension of the system, and $\epsilon$ is the error in the final state, which we call precision. Our algorithm can be directly applied as a subroutine for unitary implementation and quantum linear systems solvers, achieving $\widetilde{O}(\sqrt{N})$ dependence for both applications.
\end{abstract}

\input{introduction}

\input{datastructure}

\input{lcu}

\input{discussion}

\section{Acknowledgement}
We thank Richard Cleve and Simone Severini for the discussion and comments on this project. We also thank anonymous reviewers for their valuable suggestions and comments on this paper. CW acknowledges financial support by a David R.~Cheriton Graduate Scholarship. LW acknowledges financial support by the Royal Society through a Research Fellow Enhancement Award.
\bibliographystyle{plain}
\bibliography{bibliography}

\appendix
\input{appendix}

\end{document}

%% file: introduction.tex
\section{Introduction}
\subsection{Background and main results}
Hamiltonian simulation is the problem of simulating the dynamics of quantum systems, which is the original motivation for quantum computers~\cite{feynman1982simulating,feynman1986quantum}. It has been shown to be BQP-hard and is hence conjectured not to be classically solvable in polynomial time, since such an algorithm would imply an efficient classical algorithm for any problem with an efficient quantum algorithm, including integer factorization~\cite{shor1999polynomial}.

Different input models have been considered in previous quantum algorithms for simulating Hamiltonian evolution. The \emph{local Hamiltonian} model is specified by the local terms of a Hamiltonian. The \emph{sparse-access} model for a sparse Hamiltonian $H$ is specified by the following two oracles:
    \begin{align}
      O_{S} \ket{i,j} \ket{z} &\mapsto \ket{i,j} \ket{z \oplus S_{i,j}}, \text{ and} \\
      \label{eq:oh}
      O_{H} \ket{i,j} \ket{z} &\mapsto \ket{i,j} \ket{z \oplus H_{i,j}},
    \end{align}
    where $S_{i,j}$ is the $j$-th nonzero entry of the $i$-th row and $\oplus$ denotes the bit-wise XOR. The \emph{linear combination of unitaries} (\emph{LCU}) model is specified by a decomposition of a Hamiltonian as a linear combination of unitaries and we are given the coefficients and access to each unitary.
Following the first proposal by Lloyd~\cite{lloyd1996universal} for local Hamiltonians, Aharonov and Ta-Shma gave an efficient algorithm for sparse Hamiltonians~\cite{aharonov2003adiabatic}. Subsequently, many algorithms have been proposed which improved the runtime~\cite{berry2007efficient,berry2009black,berry2017exponential,childs2010relationship,childs2012hamiltonian,poulin2011quantum,wiebe2011simulating,low2016hamiltonian,berry2016corrected,low2017hamiltonian}, mostly in the sparse-access model, and have recently culminated in a few works with optimal dependence on all (or nearly all) parameters for sparse Hamiltonians~\cite{berry2015hamiltonian,low2017optimal,low2018hamiltonian}.

While the above-mentioned input models arise naturally in many physics applications and matrix arithmetic applications (i.e., we have access to the local terms of a Hamiltonian or each entry of a Hamiltonian can be efficiently computed), in many machine learning applications, it is more convenient to work with a different input model, namely, the \emph{quantum random access memory} (\emph{qRAM}) model, where we assume that the entries of a Hamiltonian are stored in a data structure as in~\cite{kerenidis2016quantum} and we have quantum access to the memory. As we receive and process the input data, with little overhead, they can be stored in the data structure, and then the qRAM will facilitate preparing quantum states corresponding to the input data. The use of the qRAM model has been successfully demonstrated in many applications such as quantum principal component analysis~\cite{LMR14}, quantum support vector machines~\cite{RML14}, and quantum recommendation systems~\cite{kerenidis2016quantum}.

In this work, we consider the qRAM model for obtaining information about the Hamiltonian that is not necessarily sparse. Quantum access to this data structure allows us to efficiently prepare states that encode the rows of the Hamiltonian. Using the ability to prepare these states in combination with a quantum walk~\cite{berry2015hamiltonian}, we give the \emph{first} (to the best of our knowledge) Hamiltonian simulation algorithm in the qRAM model whose time complexity has $\widetilde{O}(\sqrt{N})$ dependence\footnote{In this paper, we use $\widetilde{O}(\cdot)$ to hide poly-logarithmic factors.} for non-sparse Hamiltonians of dimensionality $N$. As a subroutine of quantum linear systems solver given by~\cite{childs2017quantum}, our result directly implies a quantum linear systems solver in the qRAM model with square-root dependence on dimension and poly-logarithmic dependence on precision, which exponentially improves the precision dependence of the quantum linear systems solver by~\cite{wossnig2018quantum} with the same input model. Since solving linear systems is a fundamental procedure of many machine learning tasks, our algorithm has extensive potential applications in \emph{quantum machine learning}. 

In~\cite{kerenidis2016quantum}, a quantum algorithm for recommendation systems was introduced based on an explicit description of a data structure, which resulted in a fast quantum algorithm for estimating singular values for any real matrix.
This fast singular value estimation algorithm was used in a quantum algorithm for solving dense linear systems~\cite{wossnig2018quantum}.
The data structure in~\cite{kerenidis2016quantum, wossnig2018quantum} allows us to prepare states that correspond to row vectors.
In our algorithm for the Hamiltonian simulation problem, the main hurdle is  to efficiently prepare the states which allow for a quantum walk corresponding to $e^{-iH/\norm{H}_1}$.
These states are quite different from those in~\cite{kerenidis2016quantum,berry2009black,berry2015hamiltonian}: the states required by~\cite{berry2015hamiltonian} allow for a quantum walk corresponding to $e^{-iH/(d\maxnorm{H})}$, where $d$ is the row-sparsity of $H$, and their states can be prepared with $O(1)$ queries to the sparse-access oracle. However, the states required by our algorithm (See Eq.~\eqref{eq:mapping-ds}) is less structural and it is not known how to prepare them with poly-logarithmic cost in the sparse-access model.
In the qRAM model, we assume the entries of $H$ are stored in a data structures, which permits state preparation with time complexity (circuit depth) $O(\polylog(N))$. The precise definition of the data structure is presented in Definition~\ref{def:datastructure}.

Using the efficient state preparation procedure, we implement the linear combination of quantum walks in order to simulate the time evolution for non-sparse Hamiltonians with only poly-logarithmic dependence on precision. The main result of this work is summarized in the following theorem, which we prove in Sec.~\ref{sec:lcu}.

\begin{theorem}[Non-sparse Hamiltonian Simulation]
  \label{thm:densehamiltoniansim}
  Let $H\in\mathbb{C}^{N\times N}$ (with $N=2^n$) be a Hermitian matrix stored in the data structure as specified in Definition~\ref{def:datastructure}. There exists a quantum algorithm for simulating the evolution of $H$ for time $t$ and error $\epsilon$ with time complexity (circuit depth)
  \begin{align}
    O\left(t\pnorm{1}{H} n^2\log^{5/2}(t\pnorm{1}{H}/\epsilon)\frac{\log(t\norm{H}/\epsilon)}{\log\log(t\norm{H}/\epsilon)}\right).
  \end{align}
\end{theorem}

In this work, we use the notation $\pnorm{1}{\cdot}$ to denote the induced 1-norm (i.e.,~maximum absolute row-sum norm), defined as $\pnorm{1}{H} = \max_j\sum_{k=0}^{N-1}|H_{jk}|$; we use $\norm{\cdot}$ to denote the spectral norm, and use $\maxnorm{\cdot}$ to denote the max norm, defined as $\maxnorm{H} = \max_{i,j}|H_{jk}|$.

By the fact that $\pnorm{1}{H}\leq\sqrt{N}\norm{H}$ (see~\cite{CK10}, alternatively, a more generalized version of this relation is shown in Appendix~\ref{appendix:norms}), we immediately have the following corollary.
\begin{corollary}
  Let $H\in\mathbb{C}^{N\times N}$ (where $N=2^n$) be a Hermitian matrix stored in the data structure as specified in Definition.~\ref{def:datastructure}. There exists a quantum algorithm for simulating the evolution of $H$ for time $t$ and error $\epsilon$ with time complexity (circuit depth)
  \begin{align}
	O\left(t\sqrt{N}\norm{H}\, n^2\log^{5/2}(t\sqrt{N}\norm{H}/\epsilon)\frac{\log(t\norm{H}/\epsilon)}{\log\log(t\norm{H}/\epsilon)}\right).
  \end{align}
\end{corollary}

\noindent
\textbf{Remarks:} 
\begin{enumerate}
  \item In Theorem~\ref{thm:densehamiltoniansim}, the circuit depth scales as $\widetilde{O}(\sqrt{N})$. However, the gate complexity could be as large as $O(N^{2.5}\log^2(N))$ in general because of the addressing scheme which allows for \textit{quantum access to classical data} stored in the data structure as specified in Definition~\ref{def:datastructure}. If some structure of $H$ is promised (e.g., the entries of $H$ repeat in some pattern), the addressing scheme could be implemented efficiently.
  \item If $H$ is $d$-sparse (i.e., $H$ has at most $d$ non-zero entries in each row), the time complexity (circuit depth) of our algorithm is
    \begin{align}
    	O\left(t\sqrt{d}\norm{H}\, n^2\log^{5/2}(t\sqrt{d}\norm{H}/\epsilon)\frac{\log(t\norm{H}/\epsilon)}{\log\log(t\norm{H}/\epsilon)}\right).
    \end{align}
    This follows from Theorem~\ref{thm:densehamiltoniansim} and the fact that $\pnorm{1}{H} \leq \sqrt{d}\norm{H}$ (as shown in Appendix~\ref{appendix:norms}).
  \item The techniques of our algorithm also work for the sparse-access model where we use standard state preparation techniques to prepare the state as in Eq.~\eqref{eq:mapping-ds}. Compared to the qRAM model, the time complexity of the state preparation for Eq.~\eqref{eq:mapping-ds} in the sparse-access model incurs an additional $O(d)$ factor (for computing $\sigma_j$). Now, to simulate $H$ for time $t$ in the sparse-access model, the dependence of the time complexity on $t$, $d$, and $\norm{H}$ becomes $O(td^{1.5}\norm{H})$ as $\norm{H}_1 \leq \sqrt{d}\norm{H}$. Hence our techniques have no advantage over previous results in the sparse-access model.
\end{enumerate}

\subsection{Related work}
\label{ssec:related}

\paragraph{Hamiltonian simulation with $\norm{H}$ dependence.}
For non-sparse Hamiltonians, a suitable model is the black-box model (a special case of the sparse-access model): querying the oracle with an index-pair $\ket{i,j}$ returns the corresponding entry of the Hamiltonian $H$, i.e., $O_H$ defined in Eq.~\eqref{eq:oh}.
With access to a black-box Hamiltonian, simulation with error $\epsilon$ can then be provably performed with query complexity 
$O((\norm{H}t)^{3/2}N^{3/4}/\sqrt{\epsilon})$ for dense Hamiltonians~\cite{berry2009black},
and it was also empirically observed in~\cite{berry2009black} that for several classes of Hamiltonians, $O(\sqrt{N}\log (N))$ queries suffice.
Whether this $\widetilde{O}(\sqrt{N})$ dependence holds for all Hamiltonians was left as an open problem. After the first version of this work was made public, this open problem was almost resolved by Low~\cite{low2018hamiltonian}, where he proposed a quantum algorithm for simulating black-box Hamiltonians with time complexity $O((t\sqrt{N}\norm{H})^{1+o(1)}/\epsilon^{o(1)})$. Although our input model is stronger than the black-box model, our work distinguishes itself since the two models are still comparable in many quantum machine learning applications and our work gives a better complexity.

\paragraph{Hamiltonian simulation with $\maxnorm{H}$ dependence.} 
Since the qRAM model is stronger than the black-box model and the sparse-access model, previous quantum algorithms such as~\cite{berry2009black,berry2015hamiltonian} can be directly used to simulate Hamiltonians in the qRAM model. For a $d$-sparse Hamiltonian, the circuit depth of the black-box Hamiltonian simulation is given by $\widetilde{O}(td\maxnorm{H})$ as shown in~\cite{berry2009black, berry2015hamiltonian}. When $H$ is non-sparse, their results imply the scaling $\widetilde{O}(tN\maxnorm{H})$. In applications where $\maxnorm{H}$ is a measure of cost, our result has no advantage against theirs, as the inequality $\norm{H}\leq\sqrt{N}\maxnorm{H}$ implies $\widetilde{O}(t\sqrt{H}\norm{H}) = \widetilde{O}(tN\maxnorm{H})$. However, in the case where $\norm{H}$ is a measure of cost, such as solving linear systems~\cite{harrow2009quantum,childs2017quantum}, and black-box unitary implementation~\cite{berry2009black}, our result has a quadratic improvement in the dimensionality dependence, as the inequality $\maxnorm{H}\leq\norm{H}$ implies $\widetilde{O}(tN\maxnorm{H}) = \widetilde{O}(tN\norm{H})$.

\paragraph{Hamiltonian simulation in the qRAM model.}
Shortly after the first version of this paper was made public, Chakraborty, Gily{\'e}n, and Jeffery~\cite{CGJ19} independently proposed a quantum algorithm for simulating non-sparse Hamiltonians based on the similar qRAM input model which achieved the same time complexity as our method. Their work is based on a very general input model, namely, the block-encoding model, which was originally proposed in~\cite{low2016hamiltonian}. The block-encoding model assumes we are given a unitary $\bigl(\begin{smallmatrix}H/\alpha & \cdot\\\cdot&\cdot\end{smallmatrix}\bigr)$ that contains $H/\alpha$ in its upper-left block. Then the evolution $e^{-iHt}$ can be simulated in $\widetilde{O}(\alpha\norm{H}t)$ time. It was shown in~\cite{low2016hamiltonian} that working with the sparse-access model of a $d$-sparse Hamiltonian $H$, a block-encoding of $H$ with $\alpha=d$ can be efficiently implemented, and it hence implies a simulation algorithm with time complexity $\widetilde{O}(d\norm{H}t)$. One of the main results in~\cite{CGJ19} is that working with the qRAM model of a $d$-sparse Hamiltonian $H$, a block-encoding with $\alpha = \sqrt{d}$ can be efficiently implemented, which yields a simulation algorithm with time complexity $\widetilde{O}(\sqrt{d}\norm{H}t)$. This result has the same complexity as ours, but their techniques are more general. In fact, the techniques of~\cite{CGJ19} has been generalized in~\cite{GSLW19} to a quantum framework for implementing singular value transformation of matrices. Another merit of~\cite{CGJ19} is that they gave detailed analysis for application to a quantum linear systems solvers in the qRAM model.

\paragraph{Quantum-inspired classical algorithms for Hamiltonian simulation.}
Recently (after the first version of this paper was made public), several classical algorithms~\cite{RWC+18,CGL+19} have been proposed for simulating Hamiltonians in a classical input model, namely, the \emph{sampling and query access}, which is comparable to the qRAM model we use. The algorithm in~\cite{RWC+18} is efficient in the low-rank and sparse regime, and the algorithm in~\cite{CGL+19} is efficient when $H$ is low-rank: their time complexity scales as $\poly(t, \norm{H}_F, 1/\epsilon)$, where $\norm{H}_F$ is the Frobenius norm of $H$. While these results have ruled out the possibility of exponential speedups of our algorithm in the low-rank regime, noting that the degrees in the polynomials of the time complexity of these classical algorithms are large, our work still has a polynomial speedup over classical algorithm (for low-rank Hamiltonians and dense Hamiltonians).

To summarize this subsection, we provide Table~\ref{tab:relatedwork} for state-of-the-art algorithms for Hamiltonian simulation in different models and how our techniques compare with theirs in the qRAM model.
\begin{table}[htbp]
  \centering
  \begin{tabular}{|l|r|r|}
\hline
\parbox[c][3em]{3.3cm}{Model} & State-of-the-art & \parbox[c][3em]{4.3cm}{Advantage of our results} \\ \hline
\parbox[c][3em]{3.3cm}{Sparse-access with \\ $\maxnorm{H}$ dependence} & $\widetilde{O}(td\maxnorm{H})$~\cite{berry2015hamiltonian} & No advantage \\ \hline
\parbox[c][3em]{3.3cm}{Sparse-access with \\ $\norm{H}$ dependence} & $O((t\sqrt{d}\norm{H})^{1+o(1)}/\epsilon^{o(1)})$~\cite{low2018hamiltonian} & \parbox[c][3em]{5.6cm}{Subpolynomial improvement in $t, d$; exponential improvement in $\epsilon$} \\ \hline
\parbox[c][2.5em]{3.3cm}{qRAM} & $\widetilde{O}(t\sqrt{d}\norm{H})$~\cite{CGJ19} & Same result \\ \hline
\parbox[c][3em]{3.3cm}{Classical sampling \\ and query access} & $\poly(t, \norm{H}_F, 1/\epsilon)$~\cite{CGL+19} & Polynomial speedup \\ \hline
\end{tabular}
\caption{Comparing our result $O(t\sqrt{d}\norm{H}\,\polylog(t,d,\norm{H}, 1/\epsilon))$ with other quantum and classical algorithms for different models. Since the qRAM model is stronger than the sparse-access model and the classical sampling and query access model, we consider the advantage of our algorithm against others when they are directly applied to the qRAM model.}
\label{tab:relatedwork}
\end{table}

\subsection{Applications}
\label{ssec:applications}
\paragraph{Unitary implementation.} One immediate application of simulating non-sparse Hamiltonians is the unitary implementation problem: given access to the entries of a unitary $U$, the objective is to construct a quantum circuit to approximate $U$ with precision $\epsilon$. As proposed in~\cite{berry2009black,jordan2009efficient}, unitary implementation can be reduced to Hamiltonian simulation by considering the Hamiltonian
\begin{align}
  H = \begin{pmatrix}
    0 & U\\
    U^{\dag} & 0
  \end{pmatrix},
\end{align}
and the fact that $e^{-iH\pi/2}\ket{1}\ket{\psi} = -i\ket{0}U\ket{\psi}$. 
If the entries of a unitary matrix can be accessed by a black-box query oracle, it is shown in~\cite{berry2009black} that this unitary operator can be implemented with $O(N^{2/3}\polylog(N)\poly(1/\epsilon))$ queries to the black box. 
Assume now that the entries of $U$ are stored in a data structure such as in Definition~\ref{def:datastructure}. 
Using the same reduction to Hamiltonian simulation, our Hamiltonian simulating algorithm implies an implementation of $U$ with time complexity (circuit depth) $O(\sqrt{N}\,\polylog(N,1/\epsilon))$.

\paragraph{Quantum linear systems solver.} Simulating non-sparse Hamiltonians can also be used as a subroutine for solving linear systems of equations for non-sparse matrices.
The essence for solving a linear system $A\ket{x}=\ket{b}$ is to apply $A^{-1}$ on $\ket{b}$, assuming here for simplicity that $\ket{b}$ is entirely in the column-space of $A$.
When $A$ is $d$-sparse, it is shown in~\cite{childs2017quantum} that $A^{-1}$ can be approximated as a linear combination of unitaries of the form $e^{-iAt}$.
An efficient quantum algorithm for Hamiltonian simulation such as~\cite{berry2015hamiltonian} can then be used as a subroutine, so that this linear system can be solved with gate complexity $O(d\kappa^2\polylog(N, \kappa/\epsilon))$.
If $A$ is non-sparse, the time complexity of their algorithm scales as $\widetilde{O}(N)$. Based on the data structure in~\cite{kerenidis2016quantum}, a quantum algorithm for solving linear systems for non-sparse matrices was described in~\cite{wossnig2018quantum}, with time complexity (circuit depth) $O(\kappa^2\sqrt{N}\polylog(N)/\epsilon)$.
If we assume a similar input model as in~\cite{kerenidis2016quantum,wossnig2018quantum}, using our Hamiltonian simulation algorithm, together with the linear combinations of unitaries (LCU) decompositions in~\cite{childs2017quantum}, we have a quantum algorithm for solving linear systems for non-sparse matrices with time complexity (circuit depth) $O(\kappa^2\sqrt{N}\polylog(\kappa/\epsilon))$, which is an exponential improvement in error dependence compared to~\cite{wossnig2018quantum}.

In the remainder of this paper, we first define our data structure, then describe the algorithm in detail and finally prove the main results. We then finish with a summary of this work and a discussion of the benefits and possible drawbacks of our algorithm.

%% file: datastructure.tex
\section{Data structure and quantum walk}
\label{sec:datastructure}

We first give the precise definition of the data structure that stores the entries of the Hamiltonian, and then show how this data structure can be used to prepare states that will allow us to perform fast Hamiltonian simulation even for dense matrices.
The data structure was introduced in in~\cite{kerenidis2016quantum}.

\begin{definition}[Data Structure]
  \label{def:datastructure}
  Let $H\in\mathbb{C}^{N\times N}$ be a Hermitian matrix (where $N = 2^n$), $\pnorm{1}{H}$ being the maximum absolute row-sum norm, and $\sigma_j:= \sum_k \lvert H_{jk} \rvert$. Each entry $H_{jk}$ is represented with $b$ bits of precision. Define $D$ as an array of $N$ binary trees $D_j$ for $j \in \{0, \ldots, N-1\}$. Each $D_j$ corresponds to the row $H_j$, and its organization is specified by the following rules. 
  \begin{enumerate}
    \item The leaf node $k$ of the tree $D_j$ stores the value\footnote{Note that the conjugation here is necessary. See Eq.~\eqref{eq:Hij}.} ${H_{jk}^{*}}$ corresponding to the index-entry pair $(j,k,H_{jk})$.
	\item For the level immediately above the bottom level, i.e., the leaves, and any node level above the leaves, the data stored is determined as follows: suppose the node has two children storing data $a$, and $b$ respectively (note that $a$ and $b$ are complex numbers). Then the entry that is stored in this node is given by $(|a|+|b|)$.
  \end{enumerate}
\end{definition}
An example of the above data structure is shown Fig.~\ref{fig:bt-4}. Note that for each binary tree $D_j$ in the data structure, the value stored in an internal (non-leaf) node is a real number, while for a leaf node, the value stored is a complex number. The root node of $D_j$ stores the value $\sum_{k=0}^{N-1}|H_{jk}^{*}|$ and we can calculate the value $\pnorm{1}{H}-\sigma_j$ in constant time. In addition, $\pnorm{1}{H}$ can be computed by taking the maximum value of the roots of all the binary trees, which can be done during the construction of the data structure.

\begin{figure}[ht]
    \centering
    \includegraphics[width=0.8\textwidth]{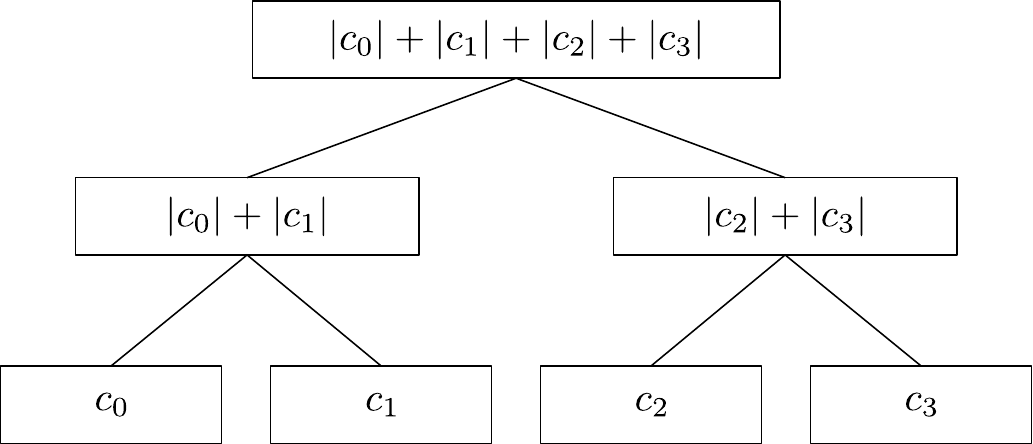}
	\caption{A example of the data structure that allows for efficient state preparation using a logarithmic number of conditional rotations.}
    \label{fig:bt-4}
\end{figure}

Using this data structure, we can efficiently perform the mapping described in the following technical lemma for efficient state preparation.

\begin{lemma}[State Preparation]
  \label{lemma:statepreparation}
  Let $H\in\mathbb{C}^{N\times N}$ be a Hermitian matrix (where $N = 2^n$) stored in the data structure as specified in Definition~\ref{def:datastructure}. Each entry $H_{jk}$ is represented with $b$ bits of precision. Then the following holds
  \begin{enumerate}
	\item Let $\pnorm{1}{H}$ be defined as above. A quantum computer that has access to the data structure can perform the following mapping for $j\in\{0,\ldots,N-1\}$,

	  \begin{align}
        \label{eq:mapping-ds}
        \ket{j}\ket{0^{\log N}}\ket{0} \mapsto \frac{1}{\sqrt{\pnorm{1}{H}}}\ket{j}\sum_{k=0}^{N-1}\ket{k}\left(\sqrt{H_{jk}^{*}}\ket{0} + \sqrt{\frac{\pnorm{1}{H} - \sigma_j}{N}}\ket{1}\right), 	  \end{align}
	  with time complexity (circuit depth) $O(n^2b^{5/2}\log b)$, where $\sigma_j = \sum_k|H_{jk}|$, and the square-root satisfies $\sqrt{H_{jk}}\bigl(\sqrt{H_{jk}^*}\bigr)^* = H_{jk}$.
	\item The size of the data structure containing all $N^2$ complex entries is $O(N^2\log^2(N))$.
  \end{enumerate}
\end{lemma}
In order to perform the mapping we will need the following Lemma, which enables us to efficiently implement the conditional rotations with complex numbers.
\begin{lemma}
  \label{lemma:conditional-rotation}
  Let $\theta, \phi_0, \phi_1 \in \mathbb{R}$ and let $\widetilde{\theta}, \widetilde{\phi_0}, \widetilde{\phi_1}$ be the $b$-bit finite precision representation of $\theta, \phi_0$, and $\phi_1$, respectively. Then there exists a unitary $U$ that performs the following mapping:
  \begin{align}
	U:\ket{\widetilde{\phi_0}}\ket{\widetilde{\phi_1}}\ket{\widetilde{\theta}}\ket{0} \mapsto \ket{\widetilde{\phi_0}}\ket{\widetilde{\phi_1}}\ket{\widetilde{\theta}}\left(e^{i\widetilde{\phi_0}}\cos(\widetilde{\theta})\ket{0} + e^{i\widetilde{\phi_1}}\sin(\widetilde{\theta})\ket{1}\right).
  \end{align}
  Moreover, $U$ can be implemented with $O(b)$ 1- and 2-qubit gates.
\end{lemma}
\begin{proof}
  Define $U$ as
  \begin{align}
	U = \left(\sum_{\widetilde{\phi_0}\in\{0,1\}^b}\ketbra{\widetilde{\phi_0}}{\widetilde{\phi_0}}\otimes e^{i\ketbra{0}{0}\widetilde{\phi_0}}\right)
	\left(\sum_{\widetilde{\phi_1}\in\{0,1\}^b}\ketbra{\widetilde{\phi_1}}{\widetilde{\phi_1}}\otimes e^{i\ketbra{1}{1}\widetilde{\phi_1}}\right)
	\left(\sum_{\widetilde{\theta}\in\{0,1\}^b}\ketbra{\widetilde{\theta}}{\widetilde{\theta}}\otimes e^{-iY\widetilde{\theta}}\right),
  \end{align}
  where $Y = \bigl(\begin{smallmatrix}0&-i\\i&0\end{smallmatrix}\bigr)$ is the Pauli $Y$ matrix.

  To implement the operator $\sum_{\widetilde{\theta}\in\{0,1\}^b}\ketbra{\widetilde{\theta}}{\widetilde{\theta}}\otimes e^{-iY\widetilde{\theta}}$, we use one rotation controlled on each qubit of the first register, with the rotation angles halved for each successive bit. The other two factors of $U$ can be implemented in a similar way. Therefore, $U$ can be implemented with $O(b)$ 1- and 2-qubit gates.
\end{proof}

Before proving Lemma~\ref{lemma:statepreparation}, we first describe the construction and the size of the data structure. Readers may refer to~\cite{kerenidis2016quantum} for more details of this data structure.
  \begin{itemize}
  \item The data structure is built from $N$ binary trees $D_i, i \in \{0, \dots, N-1\}$ and we start with an empty tree.
\item When a new entry $(i,j,H_{ij})$ arrives, we create or update the leaf node $j$ in the tree $D_i$, where the adding of the entry takes $ O(\log(N))$ time, since the depth of the tree for $H \in  \mathbb C^{N \times N}$ is at most $\log(N)$. Since the path from the root to the leaf is of length at most $\log(N)$ (under the assumption that $N=2^n$), we have furthermore to update at most $\log(N)$ nodes, which can be done in $O(\log(N))$ time if we store an ordered list of the levels in the tree. 
\item The total time for updating the tree with a new entry is given by $\log(N)\times\log(N)= \log^2(N)$.
\item The memory requirements for $k$ entries are given by $O(k \log^2(N))$ as for every entry $(j,k,H_{jk})$ at least $\log(N)$ nodes are added and each node requires at most $O(\log(N))$ bits. 
\end{itemize}

Now we are ready to prove Lemma~\ref{lemma:statepreparation}.
\begin{proof}[Proof of Lemma~\ref{lemma:statepreparation}]
With this data structure, we can perform the mapping specified in Eq.~\eqref{eq:mapping-ds}, with the following steps. For each $j$, we start from the root of $D_j$. Starting with the initial state $\ket{j}\ket{0^{\log N}}\ket{0}$, first apply the rotation (according to the value stored in the root node and calculating the normalisation in one query) on the last register to obtain the state
  \begin{align}
    \frac{1}{\sqrt{\pnorm{1}{H}}}\ket{0}\ket{0^{\log N}}\left(\sqrt{\sum_{k=0}^{N-1}|H_{jk}^{*}|}\ket{0} + \sqrt{\pnorm{1}{H} - \sigma_j}\ket{1}\right).
  \end{align}
  Then a sequence of conditional rotations is applied on each qubit of the second register to obtain the state as in Eq.~\eqref{eq:mapping-ds}. At level $\ell$ of the binary tree $D_j$, a query to the data structure is made to load the data $c$ (stored in the node) into a register in superposition, the rotation to perform is proportional to $\bigl(\sqrt{c}, \sqrt{(\pnorm{1}{H}-\sigma_j)/2^\ell}\bigr)$ (assuming at the root, $\ell = 0$, and for the leaves, $\ell = \log N$). Then the rotation angles will be determined by calculating the square root and trigonometric functions on the output of the query: this can be implemented with $O(b^{5/2})$ 1- and 2-qubit gates using simple techniques based on Taylor series and long multiplication as in~\cite{berry2015hamiltonian}, where the error is smaller than that caused by truncating to $b$ bits. Then the conditional rotation is applied by the circuit described in Lemma~\ref{lemma:conditional-rotation}, and the cost for the conditional rotation is $O(b)$. There are $n=\log(N)$ levels, so the cost excluding the implementation of the oracle is $O(nb^{5/2})$. To obtain quantum access to the classical data structure, a quantum addressing scheme is required. One addressing scheme described in~\cite{giovannetti2008quantum} can be used. Although the circuit size of this addressing scheme is $\widetilde{O}(N)$ for each $D_j$, its circuit depth is $O(n)$. Therefore, the time complexity (circuit depth) for preparing the state in Eq.~\eqref{eq:mapping-ds} is $O(n^2b^{5/2}\log n)$.
  
  We use the following rules to determine the sign of the square-root of a complex number: if $H_{jk}$ is not a negative real number, we write $H_{jk} = re^{i\varphi}$ (for $r\geq 0$ and $-\pi\leq\varphi\leq\pi$) and take $\sqrt{H_{jk}^*} = \sqrt{r}e^{-i\varphi/2}$; when $H_{jk}$ is a negative real number, we take $\sqrt{H_{jk}^*} = \mathrm{sign}(j-k)i\sqrt{|H_{jk}|}$ to avoid the sign ambiguity. With this recipe, we have $\sqrt{H_{jk}}\bigl(\sqrt{H_{jk}^*}\bigr)^* = H_{jk}$.
\end{proof}

In the following, we demonstration a state preparation procedure based on the data structure in Fig.~\ref{fig:bt-4}.
In this example, the initial state (omitting the first register) is $\ket{00}\ket{0}$. Let $\sigma_j = |c_0|+|c_1|+|c_2|+|c_3|$. Apply the first rotation, we obtain the state
\begin{align}
  \frac{1}{\sqrt{\pnorm{1}{H}}}\ket{00}\left(\sqrt{|c_0|+|c_1|+|c_2|+|c_3|}\ket{0} + \sqrt{\pnorm{1}{H}-\sigma_j}\ket{1}\right) = \nonumber \\
  \frac{1}{\sqrt{\pnorm{1}{H}}}\left(\sqrt{|c_0|+|c_1|+|c_2|+|c_3|}\ket{00}\ket{0} + \sqrt{\pnorm{1}{H}-\sigma_j}\ket{00}\ket{1}\right).
\end{align}
Then, apply a rotation on the first qubit of the first register conditioned on the last register, we obtain the state
\begin{align}
  \frac{1}{\sqrt{\pnorm{1}{H}}}\left(\left(\sqrt{|c_0|+|c_1|}\ket{00}+\sqrt{|c_2|+|c_3|}\ket{10}\right)\ket{0} + \left(\sqrt{\frac{\pnorm{1}{H}-\sigma_j}{2}}(\ket{00} + \ket{10}) \right) \ket{1}\right).
\end{align}
Next, apply a rotation on the second qubit of the first register conditioned on the first qubit of the first register and last register, we obtain the desired state:
\begin{align}
  \frac{1}{\sqrt{\pnorm{1}{H}}}\left(\sqrt{c_0}\ket{00}\ket{0}+\sqrt{c_1}\ket{01}\ket{0} + \sqrt{c_2}\ket{10}\ket{0} + \sqrt{c_3}\ket{11}\ket{0} + \right. \nonumber \\
  \left. \sqrt{\frac{\pnorm{1}{H}-\sigma_j}{4}}\ket{00}\ket{1} +\sqrt{\frac{\pnorm{1}{H}-\sigma_j}{4}}\ket{01}\ket{1} + \sqrt{\frac{\pnorm{1}{H}-\sigma_j}{4}}\ket{10}\ket{1} + \sqrt{\frac{\pnorm{1}{H}-\sigma_j}{4}}\ket{11}\ket{1}\right) = \nonumber \\
  \frac{1}{\sqrt{\pnorm{1}{H}}}\left(\ket{00}\left(\sqrt{c_0}\ket{0} + \sqrt{\frac{\pnorm{1}{H}-\sigma_j}{4}}\ket{1}\right) + \ket{01}\left(\sqrt{c_1}\ket{0} + \sqrt{\frac{\pnorm{1}{H}-\sigma_j}{4}}\ket{1}\right) + \right. \nonumber \\
  \left. \ket{10}\left(\sqrt{c_2}\ket{0} + \sqrt{\frac{\pnorm{1}{H}-\sigma_j}{4}}\ket{1}\right) + \ket{11}\left(\sqrt{c_3}\ket{0} + \sqrt{\frac{\pnorm{1}{H}-\sigma_j}{4}}\ket{1}\right)\right).
\end{align}

Based on the data structure specified in Definition~\ref{def:datastructure} and the efficient state preparation in Lemma~\ref{lemma:statepreparation}, we construct a quantum walk operator for $H$ as follows. First define the isometry $T$ as
  \begin{align}
	T = \sum_{j=0}^{N-1}\sum_{b\in\{0,1\}}(\ketbra{j}{j}\otimes\ketbra{b}{b})\otimes\ket{\varphi_{jb}},
  \end{align}
  with $\ket{\varphi_{j1}} = \ket{0}\ket{1}$ and 
  \begin{align}
    \label{eq:varphi_j0}
	\ket{\varphi_{j0}} = \frac{1}{\sqrt{\pnorm{1}{H}}}\sum_{k=0}^{N-1}\ket{k}\left(\sqrt{H_{jk}^{*}}\ket{0} + \sqrt{\frac{\pnorm{1}{H} - \sigma_j}{N}}\ket{1}\right),
  \end{align}
where $\sigma_j = \sum_{k=0}^{N-1}|H_{jk}|$.
Let $S$ be the swap operator that maps $\ket{j_0}\ket{b_0}\ket{j_1}\ket{b_1}$ to $\ket{j_1}\ket{b_1}\ket{j_0}\ket{b_0}$, for all $j_0, j_1\in\{0, \ldots, N-1\}$ and $b_0, b_1\in\{0,1\}$. We observe that
\begin{align}
  \label{eq:Hij}
  \bra{j}\bra{0}T^{\dag}ST\ket{k}\ket{0} = \frac{\sqrt{H_{jk}}\left(\sqrt{H_{jk}^*}\right)^*}{\pnorm{1}{H}} = \frac{H_{jk}}{\pnorm{1}{H}},
\end{align}
where the second equality is ensured by the choice of the square-root as in the proof of Lemma~\ref{lemma:statepreparation}.
This implies that
\begin{align}
  (I\otimes\bra{0})T^{\dag}ST(I\otimes\ket{0}) = \frac{H}{\pnorm{1}{H}}.
\end{align}

The quantum walk operator $U$ is defined as
\begin{align}
  \label{eq:quantumwalk}
  U = iS(2TT^{\dag} - I).
\end{align}
A more general characterization of the eigenvalues of quantum walks is presented in~\cite{szegedy2004quantum}. Here
we give a specific proof on the relationship between the eigenvalues of $U$ and $H$ as follows.
\begin{lemma}
  Let the unitary operator $U$ be defined as in Eq.~\eqref{eq:quantumwalk}, and let $\lambda$ be an eigenvalue of $H$ with eigenstate $\ket{\lambda}$. It holds that
  \begin{align}
	U\ket{\mu_{\pm}} = \mu_{\pm}\ket{\mu_{\pm}},
  \end{align}
  where
  \begin{align}
	\ket{\mu_{\pm}} =& (T+i\mu_{\pm}ST)\ket{\lambda}\ket{0}, \mbox{ and}\\
	\label{eq:mupm}
	\mu_{\pm} =& \pm e^{\pm i \arcsin(\lambda/\pnorm{1}{H})}.
  \end{align}
\end{lemma}
\begin{proof}
  By the fact that $T^{\dag}T = I$ and 
  $(I\otimes\bra{0})T^{\dag}ST(I\otimes\ket{0}) = H/\pnorm{1}{H}$, and $(I \otimes \bra{1})T^{\dag} ST (I \otimes \ket{0}) = 0$, we have
  \begin{align}
	U\ket{\mu_{\pm}} = \mu_{\pm}T\ket{\lambda}\ket{0}+i\left(1+\frac{2\lambda i}{\pnorm{1}{H}}\mu_{\pm}\right)ST\ket{\lambda} \ket{0}.
  \end{align}
  In order for this state being an eigenstate, it must hold that
  \begin{align}
	1 + \frac{2\lambda i}{\pnorm{1}{H}}\mu_{\pm} = \mu_{\pm}^2,
  \end{align}
  and the solution is
  \begin{align}
	\mu_{\pm} = \frac{\lambda i}{\pnorm{1}{H}} \pm \sqrt{1 - \frac{\lambda^2}{\pnorm{1}{H}^2}} =\pm e^{\pm i\arcsin(\lambda/\pnorm{1}{H})}.
  \end{align}
\end{proof}

%% file: lcu.tex
\section{Linear combination of unitaries and Hamiltonian simulation}
\label{sec:lcu}

To see how to convert the quantum walk operator $U$ to Hamiltonian simulation, we first consider the generating function for the Bessel function, denoted by $J_m(\cdot)$. According to~\cite[(9.1.41)]{abramowitz1964handbook}, we have
\begin{align}
\label{eq:relation_sum_exp}
  \sum_{m=-\infty}^{\infty}J_m(z)\mu_{\pm}^m = \exp\left(\frac{z}{2}\left(\mu_{\pm}-\frac{1}{\mu_{\pm}}\right)\right) = e^{iz\lambda/\pnorm{1}{H}},
\end{align}
where the second equality follows from Eq.~\eqref{eq:mupm} and the fact that $\sin(x)=(e^{ix}-e^{-ix})/2i$. This leads to the following linear combination of unitaries:
\begin{align}
  \label{eq:precise-V}
  V_{\infty} = \sum_{m=-\infty}^{\infty}\frac{J_m(z)}{\sum_{j={-\infty}}^{\infty}J_j(z)}U^m = \sum_{m=-\infty}^{\infty}J_m(z)U^m = e^{izH/\pnorm{1}{H}},
\end{align}
where the second equality follows from the fact that $\sum_{j={-\infty}}^{\infty}J_j(z) = 1$.

Now we consider an approximation to $e^{-izH/\pnorm{1}{H}}$ in terms of Eq.~\eqref{eq:precise-V}:
\begin{align}
  \label{eq:lcu-vk}
  V_k = \sum_{m=-k}^k\frac{J_m(z)}{\sum_{j=-k}^{k}J_j(z)}U^m.
\end{align}
Here the coefficients are normalized by $\sum_{j=-k}^k J_j(z)$ so that they sum to 1. This will minimize the approximation error (see the proof of Lemma~\ref{lemma:truncate} in Appendix~\ref{appendix:lemmas}, and the normalization trick was originated in~\cite{berry2015hamiltonian}). The eigenvalues of $V_k$ are 
\begin{align}
  \sum_{m=-k}^k\frac{J_m(z)}{\sum_{j=-k}^{k}J_j(z)}\mu_{\pm}^m.
\end{align}
Note that each eigenvalue of $V_k$ does not depend on $\pm$ as $J_{-m}(z) = (-1)^mJ_m(z)$.

To bound the error in this approximation, we have the technical lemma, and the proof is shown in Appendix~\ref{appendix:lemmas}.
\begin{lemma}
  \label{lemma:truncate}
  Let $V_k$ and $V_{\infty}$ be defined as above. There exists a positive integer $k$ satisfying $k \geq |z|$ and 
  \begin{align}
    k = O\left(\frac{\log(\norm{H}/(\pnorm{1}{H}\epsilon))}{\log\log(\norm{H}/(\pnorm{1}{H}\epsilon))}\right),
  \end{align}
  such that
  \begin{align}
	\norm{V_k-V_{\infty}} \leq \epsilon.
  \end{align}
\end{lemma}

In the following, we provide technical lemmas for implementing linear combination of unitaries.
Suppose we are given the implementations of unitaries $U_0$, $U_1$, \ldots, $U_{m-1}$, and coefficients $\alpha_0, \alpha_1, \ldots, \alpha_{m-1}$. Then the unitary
\begin{align}
  V = \sum_{j=0}^{m-1}\alpha_j U_j
\end{align}
can be implemented probabilistically by the technique called linear combination of unitaries (LCU)~\cite{kothari2014efficient}. Provided $\sum_{j=0}^{m-1}|\alpha_j| \leq 2$, $V$ can be implemented with success probability $1/4$. To achieve this, we define the multiplexed-$U$ operation, which is denoted by $\text{multi-}U$, as
\begin{align}
  \text{multi-}U\ket{j}\ket{\psi} = \ket{j}U_j\ket{\psi}.
\end{align}
We summarize the probabilistic implementation of $V$ in the following lemma, whose proof is shown in Appendix~\ref{appendix:lemmas}.
\begin{lemma}
  \label{lemma:lcu-const}
  Let $\mathrm{multi}$-$U$ be defined as above. If $\sum_{j=0}^{m-1}|\alpha_j| \leq 2$, then there exists a quantum circuit that maps $\ket{0}\ket{0}\ket{\psi}$ to the state
  \begin{align}
	\frac{1}{2}\ket{0}\ket{0}\left(\sum_{j=0}^{m-1}\alpha_jU_j\ket{\psi}\right) + \frac{\sqrt{3}}{2}\ket{\Phi^{\bot}},
  \end{align}
  where $(\ketbra{0}{0}\otimes\ketbra{0}{0}\otimes I)\ket{\Phi^{\bot}} = 0$.
  Moreover, this quantum circuit uses $O(1)$ applications of $\mathrm{multi}$-$U$ and $O(m)$ 1- and 2-qubit gates.
\end{lemma}

Let $W$ be the quantum circuit in Lemma~\ref{lemma:lcu-const}, and let $P$ be the projector defined as $P = \ketbra{0}{0}\otimes\ketbra{0}{0}\otimes I$. We have
\begin{align}
  PW\ket{0}\ket{0}\ket{\psi} = \frac{1}{2}\ket{0}\ket{0}\sum_{j=0}^{m-1}\alpha_jU_j\ket{\psi}.
\end{align}
If $\sum_{j=0}^{m-1}\alpha_jU_j$ is a unitary operator, one application of the oblivious amplitude amplification operator $-W(1-2P)W^{\dag}(1-2P)W$ implements $\sum_{j=0}^{m-1}U_j$ with certainty~\cite{berry2017exponential}. However, in our application, the unitary operator $\widetilde{W}$ implements an approximation of $V_{\infty}$ in the sense that
\begin{align}
  \label{eq:approx-lcu}
  P\widetilde{W}\ket{0}\ket{0}\ket{\psi} = \frac{1}{2}\ket{0}\ket{0}V_k\ket{\psi},
\end{align}
with $\norm{V_k - V_{\infty}} \leq \epsilon$. The following lemma shows that the error caused by the oblivious amplitude amplification is bounded by $O(\epsilon)$, and the proof is given in Appendix~\ref{appendix:lemmas}.
\begin{lemma}
  \label{lemma:oaa}
  Let the projector $P$ be defined as above. If a unitary operator $\widetilde{W}$ satisfies $P\widetilde{W}\ket{0}\ket{0}\ket{\psi} = \frac{1}{2}\ket{0}\ket{0}\widetilde{V}\ket{\psi}$ where $\norm{\widetilde{V} - V} \leq \epsilon$. Then $\norm{-\widetilde{W}(I-2P)\widetilde{W}^{\dag}(I-2P)\widetilde{W}\ket{0}\ket{0}\ket{\psi} - \ket{0}\ket{0}V\ket{\psi}} = O(\epsilon)$.
\end{lemma}
  
Now we are ready to prove Theorem~\ref{thm:densehamiltoniansim}.
\begin{proof}[Proof of Theorem~\ref{thm:densehamiltoniansim}]
The proof we outline here follows closely the proof given in~\cite{berry2015hamiltonian}. The intuition of this algorithm is to divide the simulation into $O(t\pnorm{1}{H})$ segments, with each segment simulating $e^{-iH/2}$. To implement each segment, we use the LCU technique to implement $V_k$ defined in Eq.~\eqref{eq:lcu-vk}, with coefficients $\alpha_m = J_m(z)/\sum_{j=-k}^kJ_j(z)$. 
When $z=-1/2$, we have $\sum_{j=-k}^k|\alpha_j| <2$. Actually, this holds for all $|z| \leq 1/2$ because
\begin{align}
    \sum_{j=-k}^k |\alpha_{j}| &\leq \sum_{j=-k}^k \frac{\lvert J_j(z)\rvert }{1-4 \frac{|z/2|^{k+1}}{(k+1)!}} 
    \leq \sum_{j=-k}^k \frac{\left\lvert z/2\right\rvert^{|j|}}{|j|!} \left(1 - \frac{4|z/2|^{k+1}}{(k+1)!}\right)^{-1}  \nonumber \\
    &< \frac{8}{7} + \frac{16}{7} \sum_{j=1}^{\infty} \frac{1}{4^j j!} = \frac{8}{7} + \frac{16}{7} \left(\sqrt[4]{e}-1 \right)
    < 2,
\end{align}
where the first inequality follows from the fact that $\sum_{j=-k}^kJ_j(z) \geq 1-4|z/2|^{k+1}/(k+1)!$ (see~\cite{berry2015hamiltonian}), the second inequality follows from the fact that $|J_m(z)| \leq |z/2|^{|m|}/|m|!$ (see~\cite[(9.1.5)]{abramowitz1964handbook}), and the third inequality uses the assumption that $|z|\leq 1/2$.
Now, Lemmas~\ref{lemma:lcu-const} and~\ref{lemma:oaa} can be applied. By Eq.~\eqref{eq:relation_sum_exp} and using Lemma~\ref{lemma:truncate}, set 
\begin{align}
  \label{eq:k}
  k = O\left(\frac{\log(\norm{H}/(\pnorm{1}{H}\epsilon'))}{\log\log(\norm{H}/(\pnorm{1}{H}\epsilon'))}\right),
\end{align}
and we obtain a segment that simulates $e^{-iH/(2\pnorm{1}{H})}$ with error bounded by $O(\epsilon')$. Repeat the segment $O(t\pnorm{1}{H})$ times with error $\epsilon' = \epsilon/(t\pnorm{1}{H})$, and we obtain a simulation of $e^{-iHt}$ with error bounded by $\epsilon$. It suffices to take
\begin{align}
  k = O\left(\frac{\log(t\norm{H}/\epsilon)}{\log\log(t\norm{H}/\epsilon)}\right).
\end{align}

By Lemma~\ref{lemma:lcu-const}, each segment can be implemented by $O(1)$ application of $\text{multi-}U$ and $O(k)$ 1- and 2-qubit gates, as well as the cost for computing the coefficients $\alpha_m$ for $m\in\{-k,\ldots,k\}$. The cost for each $\text{multi-}U$ is $k$ times the cost for implementing the quantum walk $U$. By Lemma~\ref{lemma:statepreparation}, the state in Eq.~\eqref{eq:varphi_j0} can be prepared with time complexity (circuit depth) $O(n^2b^{5/2})$, where $b$ is the number of bit of precision. To achieve the overall error bound $\epsilon$, we choose $b=O(\log(t\pnorm{1}{H}/\epsilon))$. Hence the time complexity for the state preparation is $O(n^2\log^{5/2}(t\pnorm{1}{H}/\epsilon))$, which is also the time complexity for applying the quantum walk $U$. Therefore, the time complexity for one segment is
\begin{align}
  O\left(n^2\log^{5/2}(t\pnorm{1}{H}/\epsilon)\frac{\log(t\norm{H}/\epsilon)}{\log\log(t\norm{H}/\epsilon)}\right).
\end{align}
Considering $O(t\pnorm{1}{H})$ segments, the time complexity is as claimed.
\end{proof}

Note that the coefficients $\alpha_{-k}, \ldots, \alpha_k$ (for $k$ defined in Eq.~\eqref{eq:k}) in Lemma~\ref{lemma:lcu-const} can be classically computed using the methods in~\cite{british1960bessel,olver1964error}, and the cost is $O(k)$ times the number of bits of precision, which is $O(\log(t\norm{H}/\epsilon)$. This is no larger than the quantum time complexity.

%% file: discussion.tex
\section{Discussion}
We presented a quantum Hamiltonian simulation algorithm whose time complexity has $\widetilde{O}(\sqrt{N})$ dependence for non-sparse Hamiltonians with access to a seemingly more powerful input model.
Our technique for Hamiltonian simulation combines ideas from linear combination of quantum walks and an efficient memory model which prepares a special states to provide the improved performance. Our algorithm can be directly applied as a subroutine for the unitary implementation problem and for a quantum linear systems solver. Especially for the latter application, many quantum machine learning algorithms are based on solving linear systems in the same input model, and our algorithm implies exponentially improved error dependence over~\cite{wossnig2018quantum} for this application.

We note that the data structure in qRAM may require a large overhead in practice, if the data were not already stored in it.
Furthermore it might be hard to implement such a data structure physically due to the exponential amount of quantum resources~\cite{aaronson2015read,adcock2015advances,ciliberto2018quantum}.
However, if the Hamiltonian is highly structured (i.e., the entries repeat in some pattern), the memory model could be efficiently implemented.
Another potential point of criticism is the required error rate of such a device, as some computations will require an error rate per gate of $O(1/\mathrm{poly}(N))$ to retain a feasible error rate for applications~\cite{arunachalam2015robustness}.
Whereas, not all computations might need such low error rates~\cite{arunachalam2015robustness} and hence the feasibility of our algorithm as a subroutine in an explicit algorithm must be further validated experimentally. To summarize, our algorithm inherits many problems of Grover's search algorithm for unordered classical data, as well as many qRAM-based quantum machine learning algorithms. 

However, if the above mentioned caveats can be overcome, our algorithm still supplies a polynomial speedup over known quantum algorithms for Hamiltonian simulation and furthermore allows for a polynomial speedup in comparison with the best known classical algorithms for several practical problems, such as linear regression or linear systems.

%% file: appendix.tex
\section{The relation between $\pnorm{1}{H}$ and $\norm{H}$}
\label{appendix:norms}
We prove the following proposition.
\begin{prop}
  If $A\in\mathbb{C}^{N\times N}$ has at most $d$ non-zero entries in any row, it holds that $\pnorm{1}{A} \leq \sqrt{d}\norm{A}$.
\end{prop}
\begin{proof}
  First observe that $\norm{A}^2 \leq \sum_i \lambda_i(A^{\dag}A) = \tr{A^{\dag}A} = \norm{A}^2_F$, and furthermore we have that
  $\sum_{ij} |a_{ij}|^2 \leq d  \max_{j \in [N]} \sum_i  |a_{ij}|^2$. From this we have that $\norm{A} \leq \sqrt{d} \norm{A}_{1}$ for a $d$-sparse $A$. By~\cite[Theorem 5.6.18]{horn1990matrix}, we have that $\norm{A}_{1} \leq C_M(1,*)\norm{A}$ for $C_M(1,*) = \max_{A \neq 0} \frac{\norm{A}}{\norm{A}_{1}}$ and using the above we have $C_M(1,*) \leq \sqrt{d}$. Therefore we find that $\norm{A}_{1} \leq \sqrt{d} \norm{A}$ as desired.
\end{proof}

It immediately follows that $\pnorm{1}{A} \leq \sqrt{N}\norm{A}$ for dense $A$, i.e., $d=N$.

\section{Proofs of technical lemmas}
\label{appendix:lemmas}
\begin{proof}[Proof of Lemma~\ref{lemma:truncate}]
The proof outlined here follows closely the proof of Lemma~8 in~\cite{berry2015hamiltonian}.
Recalling the definition of $V_k$ and $V_{\infty}$, we define the weights in $V_k$ by 
\begin{align}
  \alpha_m := \frac{J_m(z)}{C_k},
\end{align}
where $C_k = \sum_{l=-k}^k J_l(z)$. The normalization here is chosen so that $\sum_m a_m =1$ which will give the best result~\cite{berry2015hamiltonian}.\\
Since
\begin{equation}
\sum_{m=-\infty}^{\infty} J_m(z) = \sum_{m=[-\infty:-k-1;k+1:\infty]} J_m(z)  + \sum_{m=-k}^{k} J_m(z) = 1 ,
\end{equation}
observe that we have two error sources. The first one comes from the truncation of the series, and the second one comes from the different renormalization of the terms which introduces an error in the first $\lvert m \rvert \leq k$ terms in the sum.
We therefore start by bounding the normalization factor $C_k$. For the Bessel-functions for all $m$ it holds that $\lvert J_m(z) \rvert \leq \frac{1}{\lvert m\lvert !}\left\lvert \frac{z}{2} \right\rvert^{|m|}$, since $J_{-m}(z) = (-1)^m J_m(z)$~\cite[(9.1.5)]{abramowitz1964handbook}.
For $\lvert m \rvert \leq k$ we can hence find the following bound on the truncated part
\begin{align}
\sum_{m=[-\infty:-k-1;k+1:\infty]} \lvert J_m(z) \rvert &= 2 \sum_{m=k+1}^{\infty} \lvert J_m(z) \rvert \leq  2 \sum_{m=k+1}^{\infty} \frac{\lvert z/2 \rvert^m}{m!} \nonumber\\
 &= 2 \frac{\lvert z/2 \rvert^{k+1}}{(k+1)!} \left(1+\frac{|z/2|}{k+2} + \frac{|z/2|^2}{(k+2)(k+3)}  + \cdots\right) \nonumber\\
 &< 2 \frac{\lvert z/2 \rvert^{k+1}}{(k+1)!} \sum_{m=k+1}^{\infty} \left(\frac{1}{2}\right)^{m-(k+1)} = \frac{ 4 \lvert z/2 \rvert^{k+1}}{(k+1)!}.
\end{align}
Since $\sum_m J_m(z) =1$, based on the normalization, we hence find that 
\begin{equation}
\sum_{m=-k}^k J_m(z) \geq \left( 1- \frac{ 4 \lvert z/2 \rvert^{k+1}}{(k+1)!} \right),
\end{equation}
which is a lower bound on the normalization factor $C_k$. Since $a_m = \frac{J_m(z)}{C_k}$, the correction is small, which implies that 
\begin{equation}
\label{eq:dep_am}
a_m = J_m(z) \left( 1 + O \left( \frac{\lvert z/2 \rvert^{k+1}}{(k+1)!} \right) \right),
\end{equation}
and we have a multiplicative error based on the renormalization.

Next we want to bound the error in the truncation before we join the two error sources.\\
From Eq.~(\ref{eq:relation_sum_exp}) we know that
\begin{equation}
e^{iz \lambda/\Lambda} - 1 = \sum_{m=-\infty}^{\infty} J_m(z) ( \mu_{\pm}^m -1),
\end{equation}
by the normalization of $\sum_m J_m(z)$. From this we can see that we can hence obtain a bound on the truncated $J_m(z)$ as follows.
\begin{equation}
\label{eq:reordering_exp_sum}
\sum_{m=-k}^{k} J_m(z) ( \mu_{\pm}^m -1) = e^{iz \lambda/\Lambda} - 1 - \sum_{\substack{m = [-\infty: -(k+1);\\ (k+1):\infty]}} J_m(z)  ( \mu_{\pm}^m -1).
\end{equation}
Therefore we can upper bound the left-hand side in terms of the exact value of $V_{\infty}$, i.e.\ $e^{iz\lambda/\Lambda}$ if we can bound the right-most term in Eq.~\eqref{eq:reordering_exp_sum}. Using furthermore the bound in Eq.~\eqref{eq:dep_am} we obtain 
\begin{equation}
\label{eq:difference1}
\sum_{m=-k}^{k} a_m(z) ( \mu_{\pm}^m -1) = \left( e^{iz \lambda/\Lambda} - 1 - \sum_{\substack{m = [-\infty: -(k+1);\\ (k+1):\infty]}} J_m(z)  ( \mu_{\pm}^m -1)\right) \left( 1 + O \left( \frac{\lvert z/2 \rvert^{k+1}}{(k+1)!} \right) \right),
\end{equation}
which reduced with $\left\lvert 2^{iz \lambda/\Lambda} -1 \right\rvert \leq \lvert z \lambda/\Lambda \rvert$ and $|z| \leq k$ to
\begin{equation}
\label{eq:difference2}
\sum_{m=-k}^{k} a_m(z) ( \mu_{\pm}^m -1) = e^{iz \lambda/\Lambda} - 1 -  O \left(\sum_{\substack{m = [-\infty: -(k+1);\\ (k+1):\infty]}} J_m(z)  ( \mu_{\pm}^m -1)\right).
\end{equation}

We can then obtain the desired bound $\norm{ V_{\infty} - V_k}$ by reordering the above equation, and using that $\sum_{m=-k}^k a_m(z)=1$ such that we have
\begin{equation}
\label{eq:final_error_bound}
\norm{ V_{\infty} - V_k} = \left\lvert \sum_{m=-k}^k a_m \mu_{\pm}^m - e^{iz \lambda/\Lambda} \right\rvert = O \left( \sum_{\substack{m = [-\infty: -(k+1);\\ (k+1):\infty]}} J_m(z)  ( \mu_{\pm}^m -1)\right).
\end{equation}
We hence only need to bound the right-hand side.\\
For $\mu_+$ we can use that $\lvert \mu_{+}^m - 1 \rvert \leq 2 \lvert m \lambda/ \Lambda \rvert =: 2 \lvert m \nu \rvert$ and obtain the bound $ 2 \frac{\lvert \nu \rvert}{k!}\left\lvert\frac{z}{2} \right\rvert^{k+1} $~\cite{berry2015hamiltonian}. For the $\mu_-$ case we need to refine the analysis and will show that the bound remains the same. 
Let $\nu := \lambda/ \Lambda$ as above. First observe that 
$J_m(z) \mu_-^m + J_{-m} (z) \mu_i^{-m} = J_{-m} (z) \mu_+^{-m} + J_m (z) \mu_+^m$, and it follows that 
\begin{align}
 \sum_{m = -\infty}^{-(k+1)} J_m(z)  ( \mu_{-}^m -1)  &+ \sum_{m = k+1}^{\infty} J_m(z)  ( \mu_{-}^m -1) \nonumber \\
 &=\sum_{m = -\infty}^{-(k+1)} J_m(z)  ( \mu_{+}^m -1)  + \sum_{m = k+1}^{\infty} J_m(z)  ( \mu_{+}^m -1).
\end{align}
Therefore we only need to treat the $\mu_+$ case.
\begin{align}
\left\lvert \sum_{\substack{m = [-\infty: -(k+1);\\ (k+1):\infty]}} J_m(z)  ( \mu_{+}^m -1) \right\rvert &\leq 2 \left\lvert \sum_{m = k+1} ^{\infty} J_m(z)  ( \mu_{+}^m -1)  \right\rvert \nonumber\\
&\leq 2 \sum_{m= k+1}^{\infty} \lvert J_m(z) \rvert  \lvert \mu_{+}^m- 1 \rvert  \nonumber\\ 
& = 2 \sum_{m= k+1}^{\infty} \frac{1}{|m|!}\left\lvert\frac{z}{2} \right\rvert^{|m|}  \lvert \mu_{+}^m- 1 \rvert \nonumber\\
& \leq 4 \sum_{m= k+1}^{\infty} \frac{1}{|m|!}\left\lvert\frac{z}{2} \right\rvert^{|m|} m \lvert \nu \rvert \nonumber\\
&<   \frac{8 \lvert \nu \rvert }{(k+1)!}\left\lvert\frac{z}{2} \right\rvert^{k+1} (k+2) .
\end{align}
Using this bound we hence obtain from Eq.~\eqref{eq:final_error_bound},
\begin{equation}
\label{eq:final_error_bound_2}
\norm{ V_{\infty} - V_k} = \left\lvert \sum_{m=-k}^k a_m \mu_{\pm}^m - e^{iz \lambda/\Lambda} \right\rvert \leq O \left( 
\frac{\lambda}{k!\ \Lambda } \left\lvert\frac{z}{2} \right\rvert^{k+1} \right) = O\left(\frac{\norm{H}(z/2)^{k+1}}{\Lambda k!}\right).
\end{equation}
In order for the above equation being upper-bounded by $\epsilon$, it suffices to choose some $k$ that is upper bounded as claimed.
\end{proof}

\begin{proof}[Proof of Lemma~\ref{lemma:lcu-const}]
  Let $s = \sum_{j=0}^{m-1}|\alpha_j|$.
  We first define the unitary operator $B$ to prepare the coefficients:
  \begin{align}
	B\ket{0}\ket{0} = \left(\sqrt{\frac{s}{2}}\ket{0} + \sqrt{1-\frac{s}{2}}\ket{1}\right)\otimes \frac{1}{\sqrt{s}}\sum_{j=0}^{m-1}\sqrt{\alpha_j}\ket{j}.
  \end{align}
  Define the unitary operator $W$ as $W = (B^{\dag}\otimes I)(I \otimes \mbox{multi-}U)(B\otimes I)$. We claim that $W$ performs the desired mapping, as
  \begin{align}
	W\ket{0}\ket{0}\ket{\psi} =& (B^{\dag}\otimes I)(I\otimes \mbox{multi-}U)(B\otimes I) \ket{0}\ket{0}\ket{\psi} \nonumber\\
	=& \frac{1}{\sqrt{2}}(B^{\dag}\otimes I)\ket{0}\sum_{j=0}^{m-1}\sqrt{\alpha_j}\ket{j}U_j\ket{\psi} + \sqrt{\frac{2-s}{2s}}(B^{\dag}\otimes I)\ket{1}\sum_{j=0}^{m-1}\sqrt{\alpha_j}\ket{j}U_j\ket{\psi} \nonumber\\
	=& \frac{1}{2}\ket{0}\ket{0}\sum_{j=0}^{m-1}\alpha_jU_j\ket{\psi} + \sqrt{\gamma}\ket{\Phi^{\bot}},
  \end{align}
  where $\ket{\Phi^{\bot}}$ is a state satisfying $(\ketbra{0}{0}\otimes\ketbra{0}{0}\otimes I)\ket{\Phi^{\bot}} = 0$, and $\gamma$ is some normalization factor.

  The number of applications of $\mbox{multi-}U$ is constant, as in the definition of $W$. To implement the unitary operator $B$, $O(m)$ 1- and 2-qubit gates suffice.
\end{proof}

\begin{proof}[Proof of Lemma~\ref{lemma:oaa}]
We have
\begin{align}
  -\widetilde{W}(I-2P)\widetilde{W}^{\dag}(I-2P)\widetilde{W}\ket{0}\ket{0}\ket{\psi} =& (\widetilde{W} + 2P\widetilde{W} - 4\widetilde{W}P\widetilde{W}^{\dag}P\widetilde{W})\ket{0}\ket{0}\ket{\psi} \nonumber\\
  =& (\widetilde{W} + 2P\widetilde{W} - 4\widetilde{W}P\widetilde{W}^{\dag}PP\widetilde{W}P)\ket{0}\ket{0}\ket{\psi} \nonumber\\
  =& \widetilde{W}\ket{0}\ket{0}\ket{\psi} + \ket{0}\ket{0}\widetilde{V}\ket{\psi} - \widetilde{W}\left(\ket{0}\ket{0}\widetilde{V}^{\dag}\widetilde{V}\ket{\psi}\right) 
\end{align}
Because $\norm{\widetilde{V}-V} \leq \epsilon$ and $V$ is a unitary operator, we have $\norm{\widetilde{V}^{\dag}\widetilde{V} - I} = O(\epsilon)$. Therefore, we have
\begin{align}
  \norm{-\widetilde{W}(I-2P)\widetilde{W}^{\dag}(I-2P)\widetilde{W}\ket{0}\ket{0}\ket{\psi} - \ket{0}\ket{0}\widetilde{V}\ket{\psi}} = O(\epsilon).
\end{align}
Thus
\begin{align}
  \norm{-\widetilde{W}(I-2P)\widetilde{W}^{\dag}(I-2P)\widetilde{W}\ket{0}\ket{0}\ket{\psi} - \ket{0}\ket{0}V\ket{\psi}} = O(\epsilon).
\end{align}
\end{proof}